\documentclass[twocolumn,
amsmath,amssymb 
]{revtex4}

\usepackage{bm}
\usepackage{graphicx}
\usepackage{amsbsy} 
\usepackage{amsmath}
\usepackage{amsfonts}
\usepackage{amsthm}
\usepackage{bm}
\usepackage{graphicx}
\usepackage{amsbsy}
\usepackage{amsmath}
\usepackage{amsfonts}
\usepackage{amsthm}
\usepackage{braket}
\usepackage{color}
\usepackage{mathrsfs}

\begin{document}

\theoremstyle{plain}
\newtheorem{theorem}{Theorem}
\newtheorem{lemma}[theorem]{Lemma}
\newtheorem{corollary}[theorem]{Corollary}
\newtheorem{conjecture}[theorem]{Conjecture}
\newtheorem{proposition}[theorem]{Proposition}

\theoremstyle{definition}
\newtheorem{definition}{Definition}

\theoremstyle{remark}
\newtheorem*{remark}{Remark}
\newtheorem{example}{Example}

\title{Genuine Secret-Sharing States}

\author{Minjin Choi}
\affiliation{Department of Mathematics and Research Institute for Basic Sciences, 
Kyung Hee University, Seoul 02447, Korea} 

\author{Soojoon Lee}
\affiliation{Department of Mathematics and Research Institute for Basic Sciences, 
Kyung Hee University, Seoul 02447, Korea}

\date{\today}

\begin{abstract}
Quantum secret sharing allows each player 
to have classical information for secret sharing 
in quantum mechanical ways. 
In this work, we construct a class of quantum states 
on which players can quantumly perform secret sharing
secure against dishonest players as well as eavesdropper. 
We here call them the genuine secret-sharing states. 
In addition, we show that 
if $N$ players share an $N$-party genuine secret-sharing state, 
then arbitrary $M$ players out of the total players can share 
an $M$-party genuine secret-sharing state 
by means of local operations and classical communication on the state.
We also define the distillable rate with respect to the genuine secret-sharing state,
and explain the connection between the distillable rate and 
the relative entropy of entanglement.
\end{abstract}

%\pacs{03.67.Dd % Quantum cryptography and communication security
%}
\maketitle

%-----------------------------------------%
%            Introduction                 %
%-----------------------------------------%

\section{Introduction}
\label{Introduction}

Secret sharing~\cite{B79,S79} is a way of allocating a secret among players, 
and a sufficient number of players must cooperate to restore the secret.
To be more specific,
in a $(k,n)$-threshold secret sharing scheme,
a dealer distributes a secret to $n$ players,
and $k$ or more players can reconstruct the secret if they collaborate,
but fewer than $k$ players cannot do so, where $k \le n$.
Because of this feature,
secret sharing can be used to deal with important and sensitive information.

We remark that quantum mechanics can provide us
with unconditionally secure secret sharing.
For example, there is a quantum protocol~\cite{HBB99} 
based on the Greenberger-Horne-Zeilinger~(GHZ) state~\cite{GHZ}
in which each player can obtain a classical bit for $(n,n)$-threshold secret sharing.
More precisely, if $N$ players including a dealer share an $N$-qubit GHZ state,
then they can carry out an $(N-1,N-1)$-threshold secret sharing through the protocol.
As a matter of fact, this secret-sharing protocol can be considered 
as an natural generalization of 
the Bell-state based quantum key distribution~(QKD) protocol~\cite{E91}.

Even though two players can securely share a secret key through the Bell state,
there is a specific class of states 
on which perfectly secure QKD can be performed. 
The states in the class are called the private states~\cite{HO05,H09}. 
In a similar way to the private states in QKD, 
we can naturally ask the following question,  
what kinds of quantum states are required
to accomplish a secure $(n,n)$-threshold secret sharing?
In Ref.~\cite{CL08}, this question has already been considered, 
and the quantum states, called the secret-sharing states, 
have been suggested as an answer to the question.
It has been shown that
if players share a secret-sharing state,
then each player can have a classical bit for secret sharing, 
which is secure against any external eavesdropper.
However, dishonest players,
who have a fatal impact on the security of secret sharing, 
were not sufficiently considered in Ref.~\cite{CL08}.
In particular, we can find secret-sharing states that
provide each legitimate player with a secret bit,
which is insecure against dishonest players, as we will see below.
In other words, secret sharing
is not in general guaranteed on the secret-sharing states.

To see this, we first look at the secret-sharing conditions 
presented in Ref.~\cite{CL08}:
(i) The probability distributions of the players' secret bits 
must be unbiased, and perfectly correlated,
that is, if we let $p_{I}$ be the probability
that the $N$ players get the random bit string $I \in \mathbb{Z}_{2}^{N}$,
then $p_{I}=1/2^{N-1}$ for $I$ with even parity and $p_{I}=0$ for $I$ with odd parity.
(ii) Any eavesdropper cannot  obtain any information about the players' secret bits.
For $1 \le i \le N$,
let $A_{i}$ be the player $\mathcal{A}_{i}$'s qubit system,
and $A'_{i}$ be the $\mathcal{A}_{i}$'s another system with arbitrary dimension.
Then it has been shown~\cite{CL08} that
$\rho_{\bold{A}\bold{A'}}$ is a quantum state for secret sharing,
that is, a secret-sharing state,
where $\bold{A}=A_{1} \cdots A_{N}$ and $\bold{A'}=A'_{1} \cdots A'_{N}$
if and only if
it is of the form
\begin{equation}
\frac{1}{2^{N-1}} 
\sum_{\substack{I,J \in \mathbb{Z}_{2}^{N} \\ \textrm{even parity}}}
\ket{I}_{\bold{A}}\bra{J} 
\otimes
U_{I}\sigma_{\bold{A'}}U^{\dagger}_{J},
\nonumber
\end{equation}
where $\sigma_{\bold{A'}}$ is an arbitrary state,
and the $U_{I}$'s are unitary operators on the system $\bold{A'}$.
Here, $\bold{A}$ and $\bold{A'}$ are called 
the secret part and the shield part of the state, respectively.

Suppose that three players share the following secret-sharing state,
\begin{align}
\label{ex1}
\ket{\Upsilon_{1}}_{\bold{A}\bold{A'}}=
\frac{1}{2}
&(\ket{000}_{\bold{A}}\ket{000}_{\bold{A'}}
+\ket{011}_{\bold{A}}\ket{000}_{\bold{A'}} \nonumber \\
&+\ket{101}_{\bold{A}}\ket{001}_{\bold{A'}}
+\ket{110}_{\bold{A}}\ket{001}_{\bold{A'}}).
\end{align}
Then we can readily see that 
if $\mathcal{A}_{3}$ is a dishonest player,
then he/she can perfectly know the other players' secret information
by measuring his/her own secret part $A_{3}$ and shield part $A'_{3}$
in the computational basis.

There also exists a secret-sharing state
on which each legitimate player has an insecure secret bit against dishonest players,
even when dishonest players do not handle their shield parts.
It can be easily seen from the following secret-sharing state,
\begin{align*}
\ket{\Upsilon_{2}}_{\bold{A}\bold{A'}}
=\frac{1}{2\sqrt{2}}
&(\ket{0000}_{\bold{A}}\ket{0000}_{\bold{A'}}
+\ket{0011}_{\bold{A}}\ket{0000}_{\bold{A'}}\\
&+\ket{0101}_{\bold{A}}\ket{0000}_{\bold{A'}}
+\ket{0110}_{\bold{A}}\ket{0000}_{\bold{A'}}\\
&+\ket{1001}_{\bold{A}}\ket{0000}_{\bold{A'}}
-\ket{1010}_{\bold{A}}\ket{0000}_{\bold{A'}}\\
&+\ket{1100}_{\bold{A}}\ket{0000}_{\bold{A'}}
-\ket{1111}_{\bold{A}}\ket{0000}_{\bold{A'}}).
\end{align*}
In this case, $\mathcal{A}_{3}$ and $\mathcal{A}_{4}$ 
can totally know the other players' measurement outcomes
by measuring their secret parts $A_{3}A_{4}$ in the Bell basis,
and they can also deceive legitimate players
as if they measure their secret parts in the computational basis.

These problems are caused 
by the lack of sufficient consideration on dishonest players 
in the secret-sharing conditions.
Hence, in this work, we modify the secret-sharing conditions
to fully cover $(n,n)$-threshold secret sharing scenarios,
and introduce a class of quantum states
on which each player can obtain a classical information for secret sharing
secure against not only eavesdropper but also dishonest players.
We call them the genuine secret-sharing~(GSS) states.
In addition, we show that if $N$ players share an $N$-party GSS state, 
local operations and classical communication~(LOCC) enables
any $M$ players out of the $N$ players to share an $M$-party GSS state.
It can be an important property in a quantum network connected by repeaters,
since if the network consists of a GSS state 
then players can share their own GSS state with properly smaller size
without providing any information to the repeaters.

Furthermore, we define the distillable rate with respect to the GSS state,
and also show that the distillable rate is upper bounded by
the relative entropy of entanglement~\cite{REE97,REE98}
between any bipartition of the total players.
By using this property, we discuss the irreducible GSS state,
which players cannot have additional information for secret sharing 
from the shield part of the state via LOCC.

Our paper is organized as follows.
We introduce the GSS state, 
and investigate its properties in Sec.~\ref{GSS state and its properties}.
We define the GSS distillable rate in Sec.~\ref{Distillable GSS Rate},
and give examples of the GSS states in Sec.~\ref{Examples}.
Finally, we conclude our work with some discussion in Sec.~\ref{Conclusion}.

%-----------------------------------------%
%       Genuine Secret-Sharing State      %
%-----------------------------------------%

\section{Genuine secret-sharing state and its properties}
\label{GSS state and its properties}

To perfectly deal with $(n,n)$-threshold secret sharing scenarios,
we should also regard dishonest players 
as internal eavesdroppers who may conspire with an external one.
Thus we modify the secret-sharing conditions in Ref.~\cite{CL08} as follows.
\begin{itemize}
\item[(i)] 
The probability distributions of all players' secret information 
must be unbiased and perfectly correlated.
\item[(ii$'$)] 
Any eavesdropper and dishonest players cannot 
get any information about the legitimate players' secret information.
\end{itemize}
Since the modified conditions include the previous ones,
the quantum states on which players can have secret information 
that satisfies the modified conditions
also have the form of the secret-sharing states,
but they must be different from the states.
The difference can be seen in the following theorem.
From now on, we let $A_{i}$ be the qudit system for all $i$
to handle more general situations.

\begin{theorem}
\label{Thm1}
$\Upsilon_{\bold{A}\bold{A'}}$ is a quantum state
on which players can obtain secret information 
that obeys the modified secret-sharing conditions
by measuring their secret parts in the computational basis
if and only if 
for any bipartite split $\{\mathcal{P}_{1},\mathcal{P}_{2}\}$
of the players with $\left|\mathcal{P}_{1}\right|\ge 2$,
the given state $\Upsilon_{\bold{A}\bold{A'}}$ can be written as
\begin{align}
\label{GSS state}
\frac{1}{d^{N-1}}
\sum_{I_{1}I_{2}, J_{1}J_{2} \in \mathfrak{S}_{N}^{0}}
&\ket{I_{1}I_{2}}_{\bold{P}_{1}\bold{P}_{2}}\bra{J_{1}J_{2}} \nonumber \\
&\otimes
\left(U_{\bold{P}'_{1}}^{I_{1}}V^{I_{2}}_{\bold{A'}}\right)
\sigma_{\bold{A'}}
\left(U_{\bold{P}'_{1}}^{J_{1}}V^{J_{2}}_{\bold{A'}}\right)^{\dagger},
\end{align}
where 
$$\mathfrak{S}_{N}^{t}\equiv 
\left\{I=i_{1}i_{2} \cdots i_{N} \in \mathbb{Z}_{d}^{N}: 
\sum_{j=1}^{N}i_{j}\equiv t \pmod{d} \right\},$$
$\bold{P}_{k}$ and $\bold{P}_{k}^{'}$
are the secret part and the shield part of $\mathcal{P}_{k}$, respectively,
$\sigma_{\bold{A'}}$ is an arbitrary state,
and the $\left\{U_{\bold{P}'_{1}}^{I_{1}}\right\}$ and $\left\{V^{I_{2}}_{\bold{A'}}\right\}$
are unitary operators on the system $\bold{P}'_{1}$ and $\bold{A'}$, respectively.
We call the state the GSS state.
\end{theorem}

\begin{proof}
We first give a proof of the forward direction.
Suppose that $\ket{\Psi}_{\bold{A}\bold{A'}E}$
is a purification of $\Upsilon_{\bold{A}\bold{A'}}$, that is,
\begin{equation}
\label{GSS state purif}
\ket{\Psi}_{\bold{A}\bold{A'}E}
=\sum_{I \in \mathbb{Z}_{d}^{N}}
\sqrt{p_{I}}\ket{I}_{\bold{A}}\ket{\Psi_{I}}_{\bold{A'}E},
\end{equation}
where $E$ is the reference system for the purification,
which can be considered as the system of the external eavesdropper.
From the condition (i), we have
$p_{I}=1/d^{N-1}$ for $I \in \mathfrak{S}_{N}^{0}$
and $p_{I}=0$ for $I \notin \mathfrak{S}_{N}^{0}$.

For a dealer $\mathcal{A}_{k}$,
the worst case is that 
the other players except one player $\mathcal{A}_{l}~(l \ne k)$ are dishonest.
In this case, by changing the order of the systems,
we can rewrite the state $\ket{\Psi}_{\bold{A}\bold{A'}E}$ as follows:
\begin{equation}
\ket{\Psi}_{A_{k}A_{l}\bold{D}\bold{A'}E}
=\frac{1}{\sqrt{d^{2}}}
\sum_{i_{k},i_{l} \in \mathbb{Z}_{d}}
\ket{i_{k},i_{l}}_{A_{k}A_{l}}\ket{\eta_{i_{k},i_{l}}}_{\bold{D}\bold{A'}E},
\nonumber
\end{equation}
where $\bold{D}$ is the secret part of the dishonest players
and
$$\ket{\eta_{i_{k},i_{l}}}_{\bold{D}\bold{A'}E}=
\frac{1}{\sqrt{d^{N-3}}}
\sum_{\xi \in \mathfrak{S}_{N-2}^{-i_{k}-i_{l}}}
\ket{\xi}_{\bold{D}}\ket{\Psi_{i_{k},i_{l},\xi}}_{\bold{A'}E}.$$
Let $i_{k}$ be $\mathcal{A}_{k}$'s measurement outcome,
then the quantum state of dishonest players and eavesdropper after the measurement becomes
$$\Upsilon_{\bold{DD'}E}^{i_{k}}
=\frac{1}{d}\sum_{i_{l}\in \mathbb{Z}_{d}}
{\rm{tr}}_{A'_{k}A'_{l}}
\ket{\eta_{i_{k},i_{l}}}_{\bold{DA'}E}\bra{\eta_{i_{k},i_{l}}},$$
where $\bold{D'}$ is the shield part of the dishonest players.
Since the eavesdropper and dishonest players 
cannot get any information about the $\mathcal{A}_{k}$'s outcome,
$\Upsilon_{\bold{DD'}E}^{i_{k}}=\Upsilon_{\bold{DD'}E}^{i'_{k}}$
for any $i_{k},i'_{k} \in \mathbb{Z}_{d}$.
We note that 
${\rm{tr}}_{A'_{k}A'_{l}}
\ket{\eta_{i_{k},i_{l}}}_{\bold{DA'}E}
\bra{\eta_{i_{k},i_{l}}}$
is written as
\begin{align*}
\frac{1}{d^{N-3}}
\sum_{\xi,\xi' \in \mathfrak{S}_{N-2}^{-i_{k}-i_{l}}}
&\ket{\xi}_{\bold{D}}\bra{\xi'} \\
&\otimes
{\rm{tr}}_{A'_{k}A'_{l}}
\ket{\Psi_{i_{k},i_{l},\xi}}_{\bold{A'}E}\bra{\Psi_{i_{k},i_{l},\xi'}}.
\end{align*}
Hence, $\Upsilon_{\bold{DD'}E}^{i_{k}}=\Upsilon_{\bold{DD'}E}^{i'_{k}}$ implies that 
if $i_{l}$ and $i'_{l}$ satisfy $i_{k}+i_{l}=i'_{k}+i'_{l} \pmod{d}$,
then
$${\rm{tr}}_{A'_{k}A'_{l}}\ket{\eta_{i_{k},i_{l}}}_{\bold{DA'}E}
\bra{\eta_{i_{k},i_{l}}}=
{\rm{tr}}_{A'_{k}A'_{l}}\ket{\eta_{i'_{k},i'_{l}}}_{\bold{DA'}E}
\bra{\eta_{i'_{k},i'_{l}}}.$$
It follows from Hughston-Jozsa-Wooters theorem~\cite{HJW}
that for $i_{k},i'_{k},i_{l},i'_{l} \in \mathbb{Z}_{d}$, 
if $i_{k}+i_{l}=i'_{k}+i'_{l} \pmod{d}$,
there is a unitary operator $\tilde{U}_{A'_{k}A'_{l}}^{i_{k},i_{l} \to i'_{k},i'_{l}}$
on the system $A'_{k}A'_{l}$ such that
\begin{equation}
\label{relation 1}
\tilde{U}_{A'_{k}A'_{l}}^{i_{k},i_{l} \to i'_{k},i'_{l}}
\ket{\Psi_{i_{k},i_{l},\xi}}_{\bold{A'}E}
=\ket{\Psi_{i'_{k},i'_{l},\xi}}_{\bold{A'}E}
\end{equation}
for all $\xi \in \mathfrak{S}_{N-2}^{-i_{k}-i_{l}}$.

Let us now divide the players into two parties,
$\mathcal{P}_{1}$ and $\mathcal{P}_{2}$,
where $\left|\mathcal{P}_{1}\right|=M\ge 2$.
Then Eq.~(\ref{GSS state purif}) can be rewritten as
\begin{equation}
\label{GSS state purif 2}
\ket{\Psi}_{\bold{P}_{1}\bold{P}_{2}\bold{A'}E}
=\frac{1}{\sqrt{d^{N-1}}}
\sum_{I_{1}I_{2} \in \mathfrak{S}_{N}^{0}}
\ket{I_{1}I_{2}}_{\bold{P}_{1}\bold{P}_{2}}
\ket{\Psi_{I_{1}I_{2}}}_{\bold{A'}E}.
\end{equation}
Here, all possible cases of secret sharing,
including the case where $\mathcal{P}_{2}$ is the party of the dishonest players,
should be considered.
Thus, by symmetry and the Eq.~(\ref{relation 1}),
it can be shown that
there are unitary operators $U_{\bold{P}'_{1}}^{I_{1}}$
and $V_{\bold{A'}}^{I_{2}}$
such that
\begin{align*}
\ket{\Psi_{I_{1}I_{2}}}_{\bold{A'}E}
&=U_{\bold{P}'_{1}}^{I_{1}}
\ket{\Psi_{I_{1}^{\alpha}I_{2}}}_{\bold{A'}E} \\
&=U_{\bold{P}'_{1}}^{I_{1}}
V_{\bold{A'}}^{I_{2}}
\ket{\Psi_{00\cdots 0}}_{\bold{A'}E}
\end{align*}
for $I_{1} \in \mathfrak{S}_{M}^{\alpha}$
and $I_{2} \in \mathfrak{S}_{N-M}^{d-\alpha}$,
where $I_{1}^{\alpha}=0 \cdots 0\alpha \in \mathfrak{S}_{M}^{\alpha}$.
For instance, if $I=i_{1}i_{2}\cdots i_{N} \in \mathfrak{S}_{N}^{0}$,
then
\begin{equation}
\label{instance}
\ket{\Psi_{I}}_{\bold{A'}E}
=U_{A'_{1}A'_{2}}^{i_{1},i_{2}}
U_{A'_{2}A'_{3}}^{j_{2},i_{3}}
\cdots
U_{A'_{N-1}A'_{N}}^{j_{N-1},i_{N}}
\ket{\Psi_{00\cdots 0}}_{\bold{A'}E},
\end{equation}
where $U_{A'_{k}A'_{l}}^{i_{k},i_{l}}=
\left(\tilde{U}_{A'_{k}A'_{l}}^{i_{k},i_{l} \to 0,i_{k}+i_{l}}\right)^{\dagger}$
and $j_{t}\equiv i_{1}+\cdots +i_{t} \pmod{d}$.
Therefore, if we let 
${\rm{tr}}_{\bold{A'}}
\left( \ket{\Psi_{00\cdots 0}}\bra{\Psi_{00\cdots 0}}\right)
=\sum_{x}\lambda_{x}\ket{\phi_{x}}_{E}\bra{\phi_{x}}$
be its spectral decomposition,
we have
\begin{equation}
\ket{\Psi_{I_{1}I_{2}}}_{\bold{A'}E}
=\sum_{x}\sqrt{\lambda_{x}}
U_{\bold{P}'_{1}}^{I_{1}}V_{\bold{A'}}^{I_{2}}
\ket{\psi_{x}}_{\bold{A'}}\ket{\phi_{x}}_{E},
\nonumber
\end{equation}
where $\{\ket{\psi_{x}}\}$ forms an orthonormal set for the system $\bold{A'}$,
and thus $\Upsilon_{\bold{A}\bold{A'}}$ has the form in Eq.~(\ref{GSS state}).

Conversely, we now assume that
for any bipartite split $\{\mathcal{P}_{1},\mathcal{P}_{2}\}$ of the players
with $\left|\mathcal{P}_{1}\right|=M\ge 2$,
the given state $\Upsilon_{\bold{A}\bold{A'}}$
is of the form in Eq.~(\ref{GSS state}).
Then it can be readily checked that
players have secret information that
satisfies the condition (i)
by measuring their secret parts in the computational basis.
It remains to show that 
the secret information satisfies the condition (ii$'$).

Suppose that $\mathcal{P}_{1}$ and $\mathcal{P}_{2}$
are parties of legitimate players and dishonest players, respectively.
Let $\sigma_{\bold{A'}}
=\sum_{x}\lambda_{x}\ket{\mu_{x}}\bra{\mu_{x}}$
be a spectral decomposition of $\sigma_{\bold{A'}}$,
and
$$\ket{\Psi_{I_{1}I_{2}}}_{\bold{A'}E}
=\sum_{x}\sqrt{\lambda_{x}}
U_{\bold{P}'_{1}}^{I_{1}}\
V_{\bold{A'}}^{I_{2}}
\ket{\mu_{x}}_{\bold{A'}}\ket{\nu_{x}}_{E},$$
where $\left\{\ket{\nu_{x}}_{E}\right\}$ 
is an orthonormal set for the system $E$.
Then we can see that
the state of the form in Eq.~(\ref{GSS state purif 2})
is a purification of $\Upsilon_{\bold{A}\bold{A'}}$.

Let $I_{1} \in \mathfrak{S}_{M}^{\alpha}$ be
the legitimate players' measurement outcome 
after measuring their secret parts,
then the eavesdropper and dishonest players'
state after the measurement becomes
\begin{align*}
\Upsilon_{\bold{P}_{2}\bold{P}'_{2}E}^{I_{1}}
=\frac{1}{d^{N-M-1}}\sum_{I_{2},I'_{2} \in \mathfrak{S}_{N-M}^{d-\alpha}}
&\ket{I_{2}}_{\bold{P}_{2}}\bra{I'_{2}}\\
\otimes
&{\rm{tr}}_{\bold{P}'_{1}}
\ket{\Psi_{I_{1}I_{2}}}_{\bold{A'}E}\bra{\Psi_{I_{1}I'_{2}}}
\end{align*}
for $2 \le M \le N-1$ and
$\Upsilon_{E}^{I_{1}}
={\rm{tr}}_{\bold{A}'}
\ket{\Psi_{I_{1}}}_{\bold{A'}E}\bra{\Psi_{I_{1}}}$ for $M=N$.
Since this state does not depend on
the unitary operator $U_{\bold{P}'_{1}}^{I_{1}}$,
$\Upsilon_{\bold{P}_{2}\bold{P}'_{2}E}^{I_{1}}=
\Upsilon_{\bold{P}_{2}\bold{P}'_{2}E}^{I'_{1}}$
for any $I_{1}, I'_{1} \in \mathfrak{S}_{M}^{\alpha}$.
This means that the legitimate players' secret information is secure 
against dishonest players and eavesdropper.
\end{proof}

As seen in Eq.~(\ref{instance}),
unitary operators on the shield part of the GSS state
can be expressed as the product of unitary operators
acting on two players' shield parts.
\begin{corollary}
\label{coro2}
For any rearranged order $l_{1}l_{2} \cdots l_{N}$,
the GSS state can be written as
\begin{equation}
\label{GSSform2}
\frac{1}{d^{N-1}}
\sum_{I, J\in \mathfrak{S}_{N}^{0}}
\ket{I}_{A_{l_{1}}A_{l_{2}} \cdots A_{l_{N}}} \bra{J}
\otimes
U_{I} \sigma_{\bold{A'}}U_{J}^{\dagger},
\end{equation}
where
$U_{I}=U_{A'_{l_{1}}A'_{l_{2}}}^{i_{l_{1}},i_{l_{2}}}
U_{A'_{l_{2}}A'_{l_{3}}}^{j_{l_{2}},i_{l_{3}}}
\cdots
U_{A'_{l_{N-1}}A'_{l_{N}}}^{j_{l_{N-1}},i_{l_{N}}}$
for some unitary operators
$U_{A'_{l_{1}}A'_{l_{2}}}^{i_{l_{1}},i_{l_{2}}},
U_{A'_{l_{2}}A'_{l_{3}}}^{j_{l_{2}},i_{l_{3}}},
\cdots
U_{A'_{l_{N-1}}A'_{l_{N}}}^{j_{l_{N-1}},i_{l_{N}}}$
with $j_{l_{t}}\equiv i_{l_{1}}+\cdots +i_{l_{t}} \pmod{d}$.
\end{corollary}

%-------------------------------------------------------%
%       Properties of Genuine Secret Sharing State      %
%-------------------------------------------------------%

Let us now investigate the properties of the GSS state.
We remark that if $N$ players share the $N$-party GHZ state, 
then any $M$ players among the total players 
can share the $M$-party GHZ state by all players' LOCC.
We can see from the following theorem that
the GSS state has the similar property.

\begin{theorem}
\label{Thm3}
Suppose that players share a GSS state $\Upsilon_{\bold{A}\bold{A'}}$.
Then for any bipartite split $\{\mathcal{P}_{1},\mathcal{P}_{2}\}$
of the players with $\left|\mathcal{P}_{1}\right|=M\ge 2$,
$\mathcal{P}_{1}$ can share a GSS state by means of LOCC.
\end{theorem}

\begin{proof}
Let us divide $\mathcal{P}_{1}$
into legitimate players' party $\mathcal{P}_{L}$
and dishonest players' party $\mathcal{P}_{D}$
with $\left|\mathcal{P}_{L}\right|=K \ge 2$.
By rearranging the order,
let $\mathcal{P}_{L}=\{\mathcal{A}_{l_{1}},\cdots,\mathcal{A}_{l_{K}}\}$,
$\mathcal{P}_{D}=\{\mathcal{A}_{l_{K+1}},\cdots,\mathcal{A}_{l_{M}}\}$,
and $\mathcal{P}_{2}=\{\mathcal{A}_{l_{M+1}},\cdots,\mathcal{A}_{l_{N}}\}$.
Since $\Upsilon_{\bold{A}\bold{A'}}$ is a GSS state,
it follows from the Corollary~\ref{coro2}
that $\Upsilon_{\bold{A}\bold{A'}}$ is written as in Eq.~(\ref{GSSform2}).

Assume that players in $\mathcal{P}_{2}$ 
measure their secret parts in the computational basis,
and have the measurement outcome 
$I_{2}=m_{l_{M+1}} \cdots m_{l_{N}} \in \mathfrak{S}_{N-M}^{\beta}$.
Then $j_{l_{K}} = d-\beta-(i_{l_{K+1}}+ \cdots +i_{l_{M}})$,
$j_{l_{M}} = d-\beta$,
and the post-measurement state becomes as follows:
\begin{align*}
\frac{1}{d^{M-1}}
&\sum_{I_{L}I_{D}, J_{L}J_{D} \in \mathfrak{S}_{M}^{d-\beta}}
\ket{I_{L}I_{D}I_{2}}_{\bold{P}_{L}\bold{P}_{D}\bold{P}_{2}}\bra{J_{L}J_{D}I_{2}}\\
&~~~~~~~~~~~\otimes
\left(U_{\bold{P}'_{L}}^{I_{L}}V^{I_{D}}_{\bold{P}'_{1}}W_{\bold{A'}}^{I_{2}}\right)
\sigma_{\bold{A'}}
\left(U_{\bold{P}'_{L}}^{J_{L}}V^{J_{D}}_{\bold{P}'_{1}}W_{\bold{A'}}^{I_{2}}\right)^{\dagger},
\end{align*}
where
$$U_{\bold{P}'_{L}}^{I_{L}}=
U_{A'_{l_{1}}A'_{l_{2}}}^{i_{l_{1}},i_{l_{2}}}
U_{A'_{l_{2}}A'_{l_{3}}}^{j_{l_{2}},i_{l_{3}}}
\cdots
U_{A'_{l_{K-1}}A'_{l_{K}}}^{j_{l_{K-1}},i_{l_{K}}},$$
$$V_{\bold{P}'_{1}}^{I_{D}}=
U_{A'_{l_{K}}A'_{l_{K+1}}}^{j_{l_{K}},i_{l_{K+1}}}
U_{A'_{l_{K+1}}A'_{l_{K+2}}}^{j_{l_{K+1}},i_{l_{K+2}}}
\cdots
U_{A'_{l_{M-1}}A'_{l_{M}}}^{j_{l_{M-1}},i_{l_{M}}},$$
and
$$W_{\bold{A'}}^{I_{2}}=
U_{A'_{l_{M}}A'_{l_{M+1}}}^{j_{l_{M}},m_{l_{M+1}}}
U_{A'_{l_{M+1}}A'_{l_{M+2}}}^{j_{l_{M+1}},m_{l_{M+2}}}
\cdots
U_{A'_{l_{N-1}}A'_{l_{N}}}^{j_{l_{N-1}},m_{l_{N}}}.$$
By tracing out the $\mathcal{P}_{2}$'s system,
we can see that
$\mathcal{P}_{1}$'s state is written as
\begin{align*}
\frac{1}{d^{M-1}}\sum_{I_{L}I_{D}, J_{L}J_{D} \in \mathfrak{S}_{M}^{d-\beta}}
&\ket{I_{L}I_{D}}_{\bold{P}_{L}\bold{P}_{D}}\bra{J_{L}J_{D}}\\
&\otimes
\left(U_{\bold{P}'_{L}}^{I_{L}}V^{I_{D}}_{\bold{P}'_{1}}\right)
\tilde{\sigma}_{\bold{P}'_{1}}
\left(U_{\bold{P}'_{L}}^{J_{L}}V^{J_{D}}_{\bold{P}'_{1}}\right)^{\dagger}
\end{align*}
for some state $\tilde{\sigma}_{\bold{P}'_{1}}$.
Hence, after one of players $\mathcal{A}_{k}$ in $\mathcal{P}_{1}$
applies unitary operator $T_{\beta}=\sum_{i}\ket{i+\beta}\bra{i}$ on his/her secret part,
their state becomes a GSS state.
\end{proof}

Since quantum networks in general require repeaters connecting players 
because of the distance limitation of quantum communication,
and players want to get a secure information 
without providing their information to repeaters,
this property can play an important role in quantum networks.
Hence, in quantum networks,
sharing a GSS state can be a goal
for secure multipartite quantum communication.

The next property is about a relation 
between the GSS state and the Holevo information~\cite{H73},
which is one of the important quantities in security analysis.
Suppose that Alice and Bob share a state $\rho_{AB}$,
and Alice measures her system.
For each Alice's measurement outcome $x$,
let $p_{x}$ and $\rho_{B}^{x}$
be its probability and Bob's resulting state after her measurement, respectively.
The Holevo information between Alice's measurement outcomes and Bob's state
is given by
$$\chi_{\rho}(x:B)=S(\bar{\rho})-\sum_{x}p_{x}S(\rho_{B}^{x}),$$
where $S$ is the von Neumann entropy and
$\bar{\rho}=\sum_{x}p_{x}\rho_{B}^{x}$.
Zero Holevo information between Alice's measurement outcomes and Bob's state, 
that is, $\chi_{\rho}(x:B)=0$
means that Bob cannot have any information about Alice's measurement outcomes.
Thus, the modified condition (ii$'$) can be replaced as follows:
For a given quantum state $\Upsilon$,
$\chi_{\Upsilon}(k_{l}:\mathcal{D}E)=0$
for any party $\mathcal{D}E$ of dishonest players and eavesdropper,
where $k_{l}$ is the the legitimate player $\mathcal{A}_{l}$'s secret information.
In other words,
the Holevo information is zero if and only if
the state satisfying the condition (i) for secret sharing is a GSS state.
For example, since $\chi(i_{1}:\mathcal{A}_{3})=1$ in Eq.~(\ref{ex1}), 
where $i_{1}$ is the $\mathcal{A}_{1}$'s measurement outcome
when he/she measures his/her secret part in the computational basis,
the secret-sharing state in Eq.~(\ref{ex1}) is not a GSS state.

\begin{theorem}
\label{Thm4}
Let $\Gamma_{\bold{A}\bold{A'}}$ be a quantum state
and $p_{I}$ be the probability that
players' measurement outcome is $I$
when they measure the system $\bold{A}$ in the computational basis.
Suppose that $p_{I}>0$ for $I \in \mathfrak{S}_{N}^{0}$
and $p_{I}=0$ for $I \notin \mathfrak{S}_{N}^{0}$.
Then $\Gamma_{\bold{A}\bold{A'}}$ is a GSS state
if and only if 
for any $\mathcal{D} \subset \{\mathcal{A}_{1},\cdots,\mathcal{A}_{N}\}$ 
with $|\mathcal{D}|=N-2$,
the Holevo information $\chi_{\Gamma}(i_{k}:\mathcal{D}E)$ equals zero,
where $i_{k}$ is the measurement outcome of 
the $\mathcal{A}_{k}$'s secret part and $\mathcal{A}_{k} \notin \mathcal{D}$. 
\end{theorem}
\begin{proof}
We show that 
$p_{I}$'s are identical for all $I \in \mathfrak{S}_{N}^{0}$
if for any $\mathcal{D} \subset \{\mathcal{A}_{1},\cdots,\mathcal{A}_{N}\}$ 
with $|\mathcal{D}|=N-2$,
$\chi_{\Gamma}(i_{k}:\mathcal{D}E)=0$,
where $i_{k}$ is the measurement outcome of $\mathcal{A}_{k} \notin \mathcal{D}$. 
If it is true, 
then Theorem~\ref{Thm1} completes the proof.
Without loss of generality, we may assume that 
$\mathcal{D}=\{\mathcal{A}_{3},\cdots,\mathcal{A}_{N}\}$,
and $\chi_{\Gamma}(i_{k}:\mathcal{D}E)=0$ for $k=1,2$.
Note that the following state is a purification of the given state:
\begin{equation}
\ket{\Psi}_{A_{1}A_{2}\bold{D}\bold{A'}E}
=\sum_{i_{1},i_{2}=0}^{d-1}
\ket{i_{1},i_{2}}_{A_{1}A_{2}}\ket{\psi_{i_{1},i_{2}}}_{\bold{D}\bold{A'}E},
\nonumber
\end{equation}
where $\bold{D}$ is the secret part of $\mathcal{D}$
and 
$$\ket{\psi_{i_{1},i_{2}}}_{\bold{D}\bold{A'}E}=
\sum_{\xi \in \mathfrak{S}_{N-2}^{-i_{1}-i_{2}}}
\sqrt{p_{i_{1},i_{2},\xi}}\ket{\xi}_{\bold{D}}\ket{\Psi_{i_{1},i_{2},\xi}}_{\bold{A'}E}.$$
When $\mathcal{A}_{1}$'s measurement outcome is $i$,
the state of $\mathcal{D}$ and eavesdropper can be described as
$$\rho_{\mathcal{D}E}^{i_{1}=i}=\frac{1}{q_{i}}
\sum_{i_{2}=0}^{d-1}{\rm{tr}}_{A'_{1}A'_{2}}
\ket{\psi_{i,i_{2}}}_{\bold{D}\bold{A'}E}\bra{\psi_{i,i_{2}}},$$
where $q_{i}=\sum_{i_{2}}\sum_{\xi \in \mathfrak{S}_{N-2}^{-i-i_{2}}}p_{i,i_{2},\xi}$.
Since the Holevo information $\chi_{\Gamma}(i_{1}:\mathcal{D}E)$ equals zero,
$\rho_{\mathcal{D}E}^{i_{1}=i}=\rho_{\mathcal{D}E}^{i_{1}=i'}$
for any $\mathcal{A}_{1}$'s measurement outcomes $i$ and $i'$.
Similarly, since $\chi_{\Gamma}(i_{2}:\mathcal{D}E)=0$,
we also have
$\rho_{\mathcal{D}E}^{i_{2}=j}=\rho_{\mathcal{D}E}^{i_{2}=j'}$
for any $\mathcal{A}_{2}$'s measurement outcomes $j$ and $j'$,
where 
$$\rho_{\mathcal{D}E}^{i_{2}=j}=\frac{1}{r_{j}}
\sum_{i_{1}=0}^{d-1}{\rm{tr}}_{A'_{1}A'_{2}}
\ket{\psi_{i_{1},j}}_{\bold{D}\bold{A'}E}\bra{\psi_{i_{1},j}}$$
with $r_{j}=\sum_{i_{1}}\sum_{\xi \in \mathfrak{S}_{N-2}^{-i_{1}-j}}p_{i_{1},j,\xi}$.
Thus, if $i+j=i'+j' \pmod{d}$,
$$\frac{p_{i,j,\xi}}{p_{i',j',\xi}}
=\frac{q_{i}}{q_{i'}}
=\frac{r_{j}}{r_{j'}}$$ 
for $\xi \in \mathfrak{S}_{N-2}^{-i-j}$.
In particular, since
$$\frac{q_{\alpha}}{q_{\beta}}=
\frac{p_{\alpha,\beta-\alpha,\xi}}{p_{\beta,0,\xi}}
=\frac{r_{\beta-\alpha}}{r_{0}}
=\frac{p_{\alpha-\beta,\beta-\alpha,\zeta}}{p_{0,0,\zeta}}
=\frac{q_{\alpha-\beta}}{q_{0}}$$
for $\xi \in \mathfrak{S}_{N-2}^{\beta}$ and $\zeta \in \mathfrak{S}_{N-2}^{0}$,
we obtain 
$$q_{\beta}=q_{\beta}\sum_{\alpha=0}^{d-1}q_{\alpha-\beta}
=q_{0}\sum_{\alpha=0}^{d-1}q_{\alpha}=q_{0}$$
for any $\beta \in \mathbb{Z}_{d}$,
that is, $q_{i}$'s are all equal.
Hence, if $i+j=i'+j' \pmod{d}$,
then $p_{i,j,\xi}=p_{i',j',\xi}$ for all $\xi \in \mathfrak{S}_{N-2}^{-i-j}$.
By symmetry, it follows that 
$p_{I}$'s are identical for all $I \in \mathfrak{S}_{N}^{0}$.
\end{proof}

Theorem~\ref{Thm4} shows that players' secret information must be unbiased 
in order to be secure against dishonest players and eavesdropper.
Therefore, the condition (ii$'$) of the modified secret-sharing conditions
includes unbiasedness of the players' secret information,
and so the unbiasedness can be omitted in the condition (i) of the modified conditions.

%---------------------------------%
%       Distillable Key rate      %
%---------------------------------%

\section{GSS Distillable Rate}
\label{Distillable GSS Rate}

In this section, we discuss the distillable rate with respect to the GSS state.
Before defining the distillable rate,
we need to consider one issue arising from the presence of dishonest players.
For instance, let us look at the 
secret-sharing state in Eq.~(\ref{ex1}) once more. 
For $\alpha, \beta \in \mathbb{C}$ with $|\alpha|^{2}+|\beta|^{2}=1$,
if we let
$\ket{\mu_{0}}=\alpha\ket{0}+\beta\ket{1}$ and
$\ket{\mu_{1}}=\beta^{*}\ket{0}-\alpha^{*}\ket{1}$,
the state can be written as
\begin{align*}
&\frac{1}{2}\ket{\mu_{0}}_{A_{1}}\ket{0}_{A_{2}}\ket{00}_{A'_{1}A'_{2}}
\left(\alpha^{*}\ket{00}_{A_{3}A'_{3}}+\beta^{*}\ket{11}_{A_{3}A'_{3}}\right)\\
&+\frac{1}{2}\ket{\mu_{0}}_{A_{1}}\ket{1}_{A_{2}}\ket{00}_{A'_{1}A'_{2}}
\left(\alpha^{*}\ket{10}_{A_{3}A'_{3}}+\beta^{*}\ket{01}_{A_{3}A'_{3}}\right)\\
&+\frac{1}{2}\ket{\mu_{1}}_{A_{1}}\ket{0}_{A_{2}}\ket{00}_{A'_{1}A'_{2}}
\left(\beta\ket{00}_{A_{3}A'_{3}}-\alpha\ket{11}_{A_{3}A'_{3}}\right)\\
&+\frac{1}{2}\ket{\mu_{1}}_{A_{1}}\ket{1}_{A_{2}}\ket{00}_{A'_{1}A'_{2}}
\left(\beta\ket{10}_{A_{3}A'_{3}}-\alpha\ket{01}_{A_{3}A'_{3}}\right).
\end{align*}
Hence, for any $\mathcal{A}_{1}$'s measurement basis 
$\left\{\ket{\mu_{0}}, \ket{\mu_{1}}\right\}$,
player $\mathcal{A}_{3}$ can get the other players' secret information
by properly measuring his/her system.
However, since this is a secret-sharing state,
it is GHZ distillable~\cite{CL08},
that is, the GHZ state can be asymptotically distilled from the state by LOCC.
Thus, if we define the distillable rate as
how many copies of the given state are required 
to asymptotically distill a GSS state through LOCC,
the completely insecure quantum state may have a strictly positive rate.

In order to avoid such an issue,
we employ the Devetak-Winter rate~\cite{DW05, IG17}. 
Let $\rho_{\mathcal{A}_{1} \cdots \mathcal{A}_{N}}$ be a given state.
We say that $\rho_{\mathcal{A}_{1} \cdots \mathcal{A}_{N}}$
has positive Devetak-Winter rate
if there is a set of measurement operations $\{\mathcal{M}_{k}\}_{1 \le k \le N}$
such that
for any $\mathcal{D} \subset \{\mathcal{A}_{1},\cdots,\mathcal{A}_{N}\}$ 
with $|\mathcal{D}|\leq N-2$,
$I\left(m_{i}:\bar{m}_{i} \right)-\chi_{\rho}\left(m_{i}:\mathcal{D}E\right)>0$,
where $I\left(X:Y \right)$ is the mutual information between $X$ and $Y$,
$\mathcal{A}_{i} \notin \mathcal{D}$, 
$m_{i}$ is the $\mathcal{A}_{i}$'s measurement outcome,
and $\bar{m}_{i}$ is the sum of the measurement outcomes of players except $\mathcal{A}_{i}$.
If we define the distillable rate 
only for quantum states with positive Devetak-Winter rate,
then we can rule out quantum states that are completely insecure.

\begin{definition}
\label{def GSS rate}
For given state $\rho_{\mathcal{A}_{1} \cdots \mathcal{A}_{N}}$
with positive Devetak-Winter rate,
let $P_{n}$ be a sequence of LOCC operations
such that $P_{n}\left(\rho^{\otimes n}\right)=\sigma_{n}$.
We call $\mathrm{P} \equiv \bigcup_{n=1}^{\infty} \left\{P_{n}\right\}$
a GSS distillation protocol of the state $\rho$
if $\lim_{n \to \infty} \| \sigma_{n}-\Upsilon_{d_{n}}\|=0$,
where $\Upsilon_{d_{n}}$ is a GSS state
that has a secret part with dimension $d_{n}^{N}$.
For given protocol $\mathrm{P}$, its rate is defined as
\begin{equation}
R(\mathrm{P})=\limsup_{n \to \infty} \frac{\log d_{n}}{n},
\nonumber
\end{equation}
and the GSS distillable rate of the state $\rho$ is given by
\begin{equation}
D_{G}(\rho)=\sup_{\mathrm{P}}R(\mathrm{P}).
\nonumber
\end{equation}
\end{definition}

We note that when players are divided into two parties, 
$\mathcal{P}_{1}$ and $\mathcal{P}_{2}$,
they can have a private state between two parties
if they share a GSS state.
Thus, we have
$$D_{G}\left(\rho_{\mathcal{A}_{1} \cdots \mathcal{A}_{N}}\right)
\le K_{D}^{\mathcal{P}_{1}:\mathcal{P}_{2}}
\left(\rho_{\mathcal{A}_{1} \cdots \mathcal{A}_{N}}\right),$$
where $K_{D}^{\mathcal{P}_{1}:\mathcal{P}_{2}}$ is the distillable key rate
between $\mathcal{P}_{1}$ and $\mathcal{P}_{2}$,
which is defined by protocols to distill private states~\cite{H09}.
In addition, since the relative entropy of entanglement~(REE)
is an upper bound of the distillable key rate~\cite{H09},
we obtain
\begin{equation}
\label{REE bound}
D_{G}\left(\rho_{\mathcal{A}_{1} \cdots \mathcal{A}_{N}}\right)
\le E_{r}^{\mathcal{P}_{1}:\mathcal{P}_{2}}
\left(\rho_{\mathcal{A}_{1} \cdots \mathcal{A}_{N}}\right),
\end{equation}
where $E_{r}^{\mathcal{P}_{1}:\mathcal{P}_{2}}$ is REE
between $\mathcal{P}_{1}$ and $\mathcal{P}_{2}$, that is,
$$E_{r}^{\mathcal{P}_{1}:\mathcal{P}_{2}}(\rho)
=\inf_{\sigma \in \mathrm{SEP}_{\mathcal{P}_{1}:\mathcal{P}_{2}}}
S(\rho \| \sigma).$$
Here, $S(\rho \| \sigma)=-S(\rho)-\textrm{Tr}\rho \log \sigma$ is the relative entropy, 
and $\mathrm{SEP}_{\mathcal{P}_{1}:\mathcal{P}_{2}}$
is the set of bipartite separable states of the system $\mathcal{P}_{1}\mathcal{P}_{2}$.

Using the GSS distillable rate,
one can define irreducible GSS state,
as irreducible private state is defined in Ref.~\cite{H09}.
For any GSS state $\Upsilon_{\bold{A}\bold{A'}}$
whose secret part is of dimension $d^{N}$,
the state is said to be irreducible
if $D_{G}\left(\Upsilon_{\bold{A}\bold{A'}}\right)=\log d$.
Then the following theorem provides a way
to check that a given GSS state is irreducible.

\begin{theorem}
Let $\Upsilon_{\bold{A}\bold{A'}}$ be a GSS state
and $\bar{\mathcal{A}}_{k}$ denote the party of players except $\mathcal{A}_{k}$.
If $\Upsilon_{\bold{A}\bold{A'}}$ is written as
\begin{align*}
\Upsilon_{\bold{A}_{k}\bar{\bold{A}}_{k}\bold{A'}}=
\frac{1}{d^{N-1}}\sum_{i_{k}I, j_{k}J \in \mathfrak{S}_{N}^{0}}
&\ket{i_{k}I}_{\bold{A}_{k}\bar{\bold{A}}_{k}}\bra{j_{k}J}\\
\otimes
&\left(U_{\bar{\bold{A}}'_{k}}^{I}V^{i_{k}}_{\bold{A'}}\right)
\sigma_{\bold{A'}}
\left(U_{\bar{\bold{A}}'_{k}}^{J}V^{j_{k}}_{\bold{A'}}\right)^{\dagger},
\end{align*}
where $\sigma_{\bold{A'}}$ is a state on the system $\bold{A'}$,
and $\bar{\bold{A}}_{k}$ and $\bar{\bold{A}}'_{k}$
are the secret part and the shield part of $\bar{\mathcal{A}_{k}}$, respectively,
then 
\begin{equation}
\label{REE bound 2}
E_{r}^{\mathcal{A}_{k}:\bar{\mathcal{A}_{k}}}\left(\Upsilon \right)
\le \log d + \frac{1}{d}\sum_{i_{k}}
E_{r}^{\mathcal{A}_{k}:\bar{\mathcal{A}_{k}}}\left(V^{i_{k}}_{\bold{A'}}
\sigma_{\bold{A'}}\left(V^{i_{k}}_{\bold{A'}}\right)^{\dagger} \right).
\end{equation}
\end{theorem}

\begin{proof}
Let us consider the unitary operator
$$W_{\bar{\bold{A}}_{k}\bar{\bold{A}}'_{k}}
\equiv\sum_{I \in \mathbb{Z}_{d}^{N-1}}
\ket{I}_{\bar{\bold{A}}_{k}}\bra{I}
\otimes \left(U_{\bar{\bold{A}}'_{k}}^{I}\right)^{\dagger}.$$
Since the REE is invariant under local unitary operation,
$E_{r}^{\mathcal{A}_{k}:\bar{\mathcal{A}}_{k}}\left(\Upsilon \right)
=E_{r}^{\mathcal{A}_{k}:\bar{\mathcal{A}}_{k}}\left(W\Upsilon W^{\dagger}\right)$.
Note that
$$W\Upsilon W^{\dagger}=
\frac{1}{d}\sum_{i_{k}, j_{k}}
\ket{i_{k}\psi_{i_{k}}}_{\bold{A}_{k}\bar{\bold{A}}_{k}}\bra{j_{k}\psi_{j_{k}}}
\otimes
V^{i_{k}}_{\bold{A'}}
\sigma_{\bold{A'}}
\left(V^{j_{k}}_{\bold{A'}}\right)^{\dagger},$$
where 
$$\ket{\psi_{i_{k}}}_{\bar{\bold{A}}_{k}}=
\frac{1}{\sqrt{d^{N-2}}}\sum_{I \in \mathfrak{S}_{N}^{-i_{k}}}
\ket{I}_{\bar{\bold{A}}_{k}}.$$
Hence, we can regard $W\Upsilon W^{\dagger}$ as a private state.
By Theorem~3 in Ref.~\cite{H09}, we can prove this theorem.
\end{proof}

Combining inequalities (\ref{REE bound}) and (\ref{REE bound 2}),
we have 
$$D_{G}\left(\Upsilon_{\bold{A}\bold{A'}}\right)
\le \log d + \frac{1}{d}\sum_{i_{k}}
E_{r}^{\mathcal{A}_{k}:\bar{\mathcal{A}_{k}}}\left(V^{i_{k}}_{\bold{A'}}
\sigma_{\bold{A'}}\left(V^{i_{k}}_{\bold{A'}}\right)^{\dagger} \right).$$
Hence, if there is a $k$ such that
$$E_{r}^{\mathcal{A}_{k}:\bar{\mathcal{A}_{k}}}\left(V^{i_{k}}_{\bold{A'}}
\sigma_{\bold{A'}}\left(V^{i_{k}}_{\bold{A'}}\right)^{\dagger} \right)=0$$
for all $i_{k}$,
the given GSS state is irreducible,
since $D_{G}\left(\Upsilon_{\bold{A}\bold{A'}}\right) \ge \log d$.

%---------------------%
%       Examples      %
%---------------------%

\section{Examples}
\label{Examples}

\begin{example}
Let us consider a quantum network consisting of only private states.
In order for each player to get a secret bit for secret sharing in the network,
all players should be connected via private states.
More precisely, assume that Alice and Bob, Bob and Charlie, 
and Charlie and Alice share the following private states:
\begin{equation}
\label{pdit AB}
\gamma_{A_{1}B_{2}A'_{1}B'_{2}}
=\frac{1}{2} \sum_{i,i'=0}^{1}
\ket{ii}_{A_{1}B_{2}}\bra{i'i'} 
\otimes
U_{i}\sigma_{A'_{1}B'_{2}}U^{\dagger}_{i'},
\nonumber
\end{equation}
\begin{equation}
\label{pdit BC}
\hat{\gamma}_{B_{1}C_{2}B'_{1}C'_{2}}
=\frac{1}{2} \sum_{j,j'=0}^{1}
\ket{jj}_{B_{1}C_{2}}\bra{j'j'} 
\otimes
\hat{U}_{j}\hat{\sigma}_{B'_{1}C'_{2}}\hat{U}^{\dagger}_{j'},
\nonumber
\end{equation}
\begin{equation}
\label{pdit CA}
\tilde{\gamma}_{C_{1}A_{2}C'_{1}A'_{2}}
=\frac{1}{2} \sum_{k,k'=0}^{1}
\ket{kk}_{C_{1}A_{2}}\bra{k'k'} 
\otimes
\tilde{U}_{k}\tilde{\sigma}_{C'_{1}A'_{2}}\tilde{U}^{\dagger}_{k'},
\nonumber
\end{equation}
where $\sigma_{A'_{1}B'_{2}}$, $\hat{\sigma}_{B'_{1}C'_{2}}$,
and $\tilde{\sigma}_{C'_{1}A'_{2}}$ are arbitrary states,
and $\left\{U_{i}\right\}$, $\left\{\hat{U}_{j}\right\}$,
and $\left\{\tilde{U}_{k}\right\}$
are unitary operators on the system $A'_{1}B'_{2}$,
$B'_{1}C'_{2}$, and $C'_{1}A'_{2}$, respectively.
If $m_{1}$ and $m_{2}$ are measurement outcomes
when they measure their key parts 1 and 2 in the computational basis, respectively, 
then $m_{1}+m_{2} \pmod{2}$ can be used as a secret bit for secret sharing.

This secret bit can also be obtained from the process 
where each player takes the unitary 
$W\equiv\sum_{i,j=0}^{1}\ket{i+j,j}\bra{i,j}$ on his/her key parts, 
and measures the first key part in the computational basis.
After taking the unitary operators on their key parts, 
the state becomes
\begin{align*}
\Upsilon&=\frac{1}{4}\sum_{ijk,i'j'k'\in\mathfrak{S}_{3}^{0}}
\ket{ijk}_{\bold{P}}\bra{i'j'k'}
\otimes
V^{i}\hat{V}^{k}\Lambda\left(V^{i'}\hat{V}^{k'}\right)^{\dagger}\\
&=\frac{1}{4}\sum_{ijk,i'j'k'\in\mathfrak{S}_{3}^{0}}
\ket{ijk}_{\bold{P}}\bra{i'j'k'}
\otimes
\hat{V}^{j}\tilde{V}^{i}\Lambda\left(\hat{V}^{j'}\tilde{V}^{i'}\right)^{\dagger}\\
&=\frac{1}{4}\sum_{ijk,i'j'k'\in\mathfrak{S}_{3}^{0}}
\ket{ijk}_{\bold{P}}\bra{i'j'k'}
\otimes
\tilde{V}^{k}V^{j}\Lambda\left(\tilde{V}^{k'}V^{j'}\right)^{\dagger},
\end{align*}
where $\bold{P}=A_{1}B_{1}C_{1}$,
\begin{align*}
\Lambda=\frac{1}{2}\sum_{l,l'=0}^{1}
&\ket{lll}_{A_{2}B_{2}C_{2}}\bra{l'l'l'}\\
&\otimes U_{l}\sigma_{A'_{1}B'_{2}}U^{\dagger}_{l'}
\otimes \hat{U}_{l}\hat{\sigma}_{B'_{1}C'_{2}}\hat{U}^{\dagger}_{l'}
\otimes \tilde{U}_{l}\tilde{\sigma}_{C'_{1}A'_{2}}\tilde{U}^{\dagger}_{l'},
\end{align*}
$$V_{B_{2}A'_{1}B'_{2}}^{1}=
\sum_{l=0}^{1}\ket{l}_{B_{2}}\bra{l+1}
\otimes U_{l}U_{l+1}^{\dagger},$$
$$\hat{V}_{C_{2}B'_{1}C'_{2}}^{1}=
\sum_{l=0}^{1}\ket{l}_{C_{2}}\bra{l+1}
\otimes \hat{U}_{l}\hat{U}_{l+1}^{\dagger},$$
$$\tilde{V}_{A_{2}C'_{1}A'_{2}}^{1}=
\sum_{l=0}^{1}\ket{l}_{A_{2}}\bra{l+1}
\otimes \tilde{U}_{l}\tilde{U}_{l+1}^{\dagger},$$
and $V^{0}$, $\hat{V}^{0}$ and $\tilde{V}^{0}$ are identity operators.
This is to prepare the GSS state $\Upsilon$
so as to obtain a secret bit for secret sharing.

One may think that sharing private states is enough for players
to carry out secret sharing.
However, as we can see above, 
it can be seen as a process of preparing a GSS state
from private states which share in advance between every pair of the total players, 
and hence it can be spatially inefficient
because the second key parts of the private states are not used 
in obtaining a secret bit.
Therefore, if it is possible to directly share a GSS state,
it can be more efficient and more productive than sharing private states
when players want to perform secret sharing.
\end{example}

\begin{example}
As in Refs.~\cite{CL08, H09},
it can be an interesting task 
to find a GSS state with low distillable entanglement.
To this end, let us first consider the following state:
\begin{align*}
\Gamma=&
a_{0}\ket{\psi_{0}}_{ABC}\bra{\psi_{0}}\otimes
\left(\rho_{0} \otimes \sigma_{0} \otimes \tau_{0}\right)\\
&+a_{1}\ket{\psi_{1}}_{ABC}\bra{\psi_{1}}\otimes
\left(\rho_{1} \otimes \sigma_{0} \otimes \tau_{1}\right)\\
&+a_{2}\ket{\psi_{2}}_{ABC}\bra{\psi_{2}}\otimes
\left(\rho_{1} \otimes \sigma_{1} \otimes \tau_{0}\right)\\
&+a_{3}\ket{\psi_{3}}_{ABC}\bra{\psi_{3}}\otimes
\left(\rho_{0} \otimes \sigma_{1} \otimes \tau_{1}\right),
\end{align*}
where $\ket{\psi_{0}}=\frac{1}{2}
\left(\ket{000}+\ket{011}+\ket{101}+\ket{110}\right)$,
$\ket{\psi_{1}}=Z_{A}\ket{\psi_{0}}$,
$\ket{\psi_{2}}=Z_{B}\ket{\psi_{0}}$,
$\ket{\psi_{3}}=Z_{C}\ket{\psi_{0}}$,
and
$\rho_{i}$, $\sigma_{j}$, and $\tau_{k}$ have orthogonal supports 
on the system $A'_{1}B'_{2}$, $B'_{1}C'_{2}$, and $C'_{1}A'_{2}$, respectively.
Then we can find unitary operators $U_{A'_{1}B'_{2}}$,
$V_{B'_{1}C'_{2}}$, and $W_{C'_{1}A'_{2}}$
that satisfy $U_{A'_{1}B'_{2}}\rho_{i}=(-1)^{i}\rho_{i}$,
$V_{B'_{1}C'_{2}}\sigma_{j}=(-1)^{j}\sigma_{j}$,
and $W_{C'_{1}A'_{2}}\tau_{k}=(-1)^{k}\tau_{k}$.
Using these operators, it can be shown that $\Gamma$ is a GSS state.

We now consider the state
\begin{align*}
\Gamma=&
p\ket{\psi_{0}}\bra{\psi_{0}}\otimes
\left(\varrho_{s} \otimes \sigma_{0} \otimes \varrho_{s}\right)\\
&+p\ket{\psi_{1}}\bra{\psi_{1}}\otimes
\left(\varrho_{a} \otimes \sigma_{0} \otimes \varrho_{a}\right)\\
&+\left(\frac{1}{2}-p\right)\ket{\psi_{2}}\bra{\psi_{2}}\otimes
\left(\varrho_{a} \otimes \sigma_{1} \otimes \varrho_{s}\right)\\
&+\left(\frac{1}{2}-p\right)\ket{\psi_{3}}\bra{\psi_{3}}\otimes
\left(\varrho_{s} \otimes \sigma_{1} \otimes \varrho_{a}\right),
\end{align*}
where $\varrho_{s}$ and $\varrho_{a}$ are
symmetric and antisymmetric Werner states~\cite{W89}
$$\varrho_{s}
=\frac{\mathcal{I}+\mathcal{F}}{d^{2}+d},$$
$$\varrho_{a}
=\frac{\mathcal{I}-\mathcal{F}}{d^{2}-d}$$
with identity operator $\mathcal{I}$ on the $d \otimes d$ system
and the flip operator $\mathcal{F}=\sum_{i,j=0}^{d-1}\ket{ij}\bra{ji}$.
It follows from tedious but straightforward calculations
that $$\left\|\Gamma^{T_{AA'_{1}A'_{2}}}\right\|_{1}=
1+\left(\frac{2}{d^{2}}+\frac{2}{d}\right)(1+2p),$$
where $\|\cdot\|_{1}$ is the trace norm.
Since the distillable entanglement is upper bounded by
the log negativity~\cite{VW02}
$E_{N}(\Gamma)=\log_{2}\left\|\Gamma^{T_{AA'_{1}A'_{2}}}\right\|_{1}$,
we can construct the GSS state
that has arbitrarily low bipartite distillable entanglement
between Alice and the rest of players by increasing $d$.
\end{example}

%-----------------------%
%       Conclusion      %
%-----------------------%

\section{Conclusion}
\label{Conclusion}
In this work, we have presented the GSS state, and explored its properties.
The GSS state provides a classical information for $(n,n)$-threshold secret sharing,
which is secure against dishonest players and eavesdropper.
Furthermore, if $N$ players share an $N$-party GSS state,
LOCC enables arbitrary $M$ players out of the total players 
to share an $M$-party GSS state. 
We have also defined the GSS distillable rate, 
and it has been shown that 
the REE between any bipartition of the total players 
is an upper bound on the GSS distillable rate.

The GSS state can be regarded as 
a generalization of the private state with respect to secret sharing.
In addition, when players share a GSS state,
any two players among them can share a private state by all players' LOCC.
Thus, by applying various research results related to the private state,
the results can also be used to investigate multipartite communication.

In a quantum network, if players share a GSS state,
any arbitrary parties of players can perform QKD or quantum secret sharing.
We can naturally have the following question:
Is the GSS state the only quantum state on which players can do them?
If we answer this question,
we can have a more in-depth discussion of the secure quantum network.

There are interesting future works related to the GSS state.
First, we can think about a bound entangled state with positive GSS distillable rate.
If we find such a state,
it can help us to study the relationship 
between distillable multipartite entanglement and the GSS distillable rate in detail.
Second, it can be an intriguing task 
to find a way to share a GSS state between players and repeaters in quantum networks.
If there exists such a way,
we can construct a quantum network
that ensures secure multipartite communication among players.
Therefore the GSS state could be considered as a new resource
in multipartite quantum communication.

\acknowledgements{
This research was supported by the Basic Science Research
Program through the National Research Foundation of Korea funded by the Ministry of Science and ICT 
(Grant No.NRF-2019R1A2C1006337) and the Ministry of Science and
ICT, Korea, under the Information Technology Research Center support program 
(Grant No. IITP-2020-2018-0-01402) supervised by the Institute for Information and Communications
Technology Promotion.
S.L. acknowledges support from Research Leave
Program of Kyung Hee University in 2018.
}

\bibliography{reference}

\end{document}